\documentclass[aps,prl,superscriptaddress,twocolumn,nolongbibliography]{revtex4-2}
\usepackage[colorlinks=true,citecolor=cyan,urlcolor=blue]{hyperref}
\usepackage{times,epsfig,amssymb,amsfonts,amsmath,bm,bbm,subfigure,
mathtools,amsthm,braket,soul,enumitem,xcolor,physics,graphics,float,
graphicx,array,adjustbox,tikz,multirow,comment}
\newtheorem{thm}{Theorem}
\newtheorem{lem}{Lemma}

\newtheorem{obs}{Observation}
\newtheorem{defn}{Definition}
\newtheorem*{defn*}{Definition}
\newcommand{\bet}{\begin{theorem}}
\newcommand{\beq}{\begin{equation}}
\newcommand{\enq}{\end{equation}}
\usepackage{xcolor,hyperref}
 \definecolor{darkblue}{rgb}{0,0.0.1,0.3}
 \definecolor{darkred}{rgb}{0.6,0.1,0}
 \hypersetup{colorlinks,breaklinks,
   linkcolor=darkred,urlcolor=darkblue,
   anchorcolor=darkred,citecolor=
   darkred,pdfauthor=GrAr, pdftitle=GrAr.UPB.Context}
\begin{document}
\title{Graph Theoretic Quantum Contextuality and Unextendible
Product Bases}
\author{Gurvir Singh}
\email{gurvirsingh@iisermohali.ac.in}
\affiliation{Department of Physical Sciences,
Indian Institute of Science Education and
Research Mohali, Punjab 140306 India}
\author{Arvind}
\email{arvind@iisermohali.ac.in}
\affiliation{Department of Physical Sciences,
Indian Institute of Science Education and
Research Mohali, Punjab 140306 India}
\begin{abstract}
Unextendible product bases(UPBs) are central to the study of
local distinguishability of orthogonal product
states.  While their connection to quantum
nonlocality via Bell inequalities is well
established, their link to quantum contextuality
remains largely unexplored. We establish a graph
theoretic connection between contextuality and UPBs.
First, an equivalence between
Klyachko-Can-Binicio\u{g}lu-Shumovsky (KCBS) vectors
and the Pyramid UPB is shown and then by
constructing a one parameter family of UPB vectors,
a quantitative connection between
\textit{`contextuality strength}' and bound
entanglement of states associated with the
corresponding UPB is demonstrated.  This equivalence
is extended to generalized KCBS vectors and the {\it
GenPyramid} UPB.  A new class of minimal UPBs in
\(\mathbb{C}^3 \otimes \mathbb{C}^n\) is constructed
using \textit{Lov\'asz-optimal} orthogonal
representations (LOORs) of cycle graphs and their
complements which we term the \textit{GenContextual}
UPB.  Any minimal UPB in this dimension is shown to
be graph-equivalent to the
\textit{GenContextual} UPB. We briefly discuss the
distinguishability properties of
\textit{GenContextual} UPB.  In the reverse
direction, we observe that the constituent
vectors of the \textit{QuadRes} UPB are LOORs of
Paley graphs. The structural properties of these
graphs make them suitable candidates for
constructing noncontextuality inequalities, thereby
establishing a bidirectional connection between
quantum contextuality and UPBs.
\end{abstract}
\maketitle 
\noindent
\textit{Introduction.--}
Graph theoretic concepts have been used as a powerful tool
in the study of quantum contextuality as well as of
Unextendable Product Bases UPBs~\cite{Csw2010}.
Cabello-Severini-Winter introduced a graph-theoretic
framework for contextuality, where measurement compatibility
relations are represented by graphs~\cite{Csw2010,Csw2014}.
This approach connects quantum contextuality to graph
theoretic quantities like the independence number and the
Lovász number~\cite{Lovasz1989}.  Graph theoretic techniques
have also been employed to express orthogonality relations
and to reformulate the unextendibility condition of
UPBs~\cite{Fei2023}. We establish an \textit{a priori} unexpected
connection between contextuality and UPBs based on
the underlying graph theoretic structure.

UPBs are intriguing objects~\cite{Bennett1999b}, as they
connect on the one hand to \textit{nonlocality without
entanglement} (NLWE)~\cite{Bennett1999a} and on the other to
bound entanglement~\cite{DiVincenzo2003}. NLWE originates
from the fact that UPBs are sets of orthogonal product
states that cannot be perfectly distinguished using local
operations and classical communication (LOCC); the bound
entanglement arises since the projector onto orthogonal
complement of UPBs defines an entangled state which is
positive under partial transpose by construction.  The
distinguishability of UPBs under various kinds of
measurements~\cite{Cohen2008,Cohen2022,Cohen2023,Bandyopadhyay2015}
has been extensively studied.

Contextuality marks a fundamental departure from classical
behavior of physical systems.  Quantum theory is contextual
as the outcomes of measurements depend on the measurement
context and therefore the existence of non-contextual hidden
variable models to explain quantum correlations is ruled
out~\cite{specker1967}.  Similar to how Bell inequality
violations indicate nonlocality, contextuality is
experimentally witnessed through violations of
noncontextuality
inequalities~\cite{Cabello2008a,Cabello2008b,Badziag2009},
the simplest of which was introduced by
KCBS~\cite{Kcbs2008}.

Connections between Bell inequalities with no quantum
violation and UPBs have been known for over a
decade~\cite{Augusiak2011}. Notably, a one to one
correspondence was established between tight Bell
inequalities with no quantum violation and certain
multipartite UPBs, leading to the construction of a class of
two qubit multipartite UPBs~\cite{Augusiak2012}. 
The derivation of these Bell inequalities is based on
the principle of local
orthogonality~\cite{Fritz2013,Sainz2014}, which in turn is a
consequence of the exclusivity principle
applied to Bell scenarios; a concept originally proposed to
explain the maximal quantum violation of a noncontextuality
inequality~\cite{Csw2010}. The exclusivity principle states that the
sum of probabilities of pairwise exclusive events cannot
exceed one, with exclusivity represented by orthogonality.
Motivated by the significant connection between Bell
nonlocality and UPBs, we explore, establish and
generalise a direct
link between contextuality and UPBs through a
graph theoretic framework.

We notice that the five vector in $\mathbb{C}^3$ that are
motifs of the Pyramid UPB are equivalent to the KCBS vectors.
Motivated by this equivalence  we construct a one-parameter family of
five contextual vectors in $\mathbb{C}^3$ with
corresponding UPBs in $\mathbb{C}^3 \otimes \mathbb{C}^3$
and show that the {\it `contextuality strength'} of these
vectors is quantitatively connected with the bound
entanglement of the states associated with the corresponding
UPBs. This family contains Pyramids and Tiles as special
cases and thus we provide an explanation of their different
linear entropy of entanglement based on contextuality.
The equivalence is further extended to the Generalized KCBS
vectors, which constitute the \textit{Lovász  optimal}
orthogonal representations (LOORs) of odd cycle graphs
$C_n$, and a class of UPBs known as {\it GenPyramid} UPBs.
While previous constructions of {\it GenPyramid} UPBs were
limited to prime dimensions~\cite{DiVincenzo2003}, we find
that such UPBs also exist in certain non prime dimensions.
Building on the LOORs of odd cycle graphs and their
complements, we introduce a new class of minimal UPBs in
$\mathbb{C}^3 \otimes \mathbb{C}^n$ for odd $n$, which we
refer to as \textit{GenContextual} UPB, and prove that any
minimal UPB in these dimensions is graph equivalent to
the corresponding \textit{GenContextual} UPB.
Further, we observe that the connection between UPBs and
contextuality can be exploited in the opposite direction; we
show that a class of minimal UPBs known as \textit{QuadRes}
UPB can be realized as LOORs of Paley graphs, identifying
them as promising candidates for further exploration of
quantum contextuality.
\begin{figure}[t!]
\centering
\includegraphics[scale=0.75]{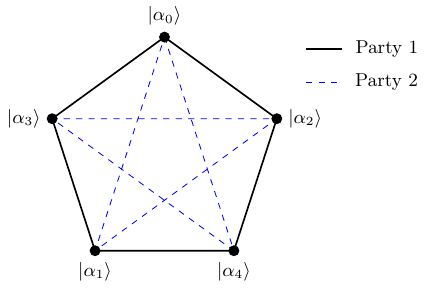}
\caption{
Orthogonality graph of UPBs in $\mathbb{C}^3 \otimes
\mathbb{C}^3$.  The KCBS vectors also exhibit the
same pentagonal ($C_5$) orthogonality structure.
Vertices label the five product states
$|\alpha_i\rangle$ and \textit{edges encode
orthogonality}.  Thick edges (Party~1) form a
pentagon $C_5$ with cycle
$0\!\to\!2\!\to\!4\!\to\!1\!\to\!3\!\to\!0$.  The
overall graph is the complete graph $K_5 = C_5 \cup
\overline{C}_5$, with Party~1 forming the pentagon
($C_5$) and Party~2 its complement $\overline{C}_5$.
Note that $C_5 \cong \overline{C}_5$.}
\label{orthogonality_upb}
\end{figure}

\noindent
\textit{Contextuality-UPB correspondence.--}
Consider a one parameter family of five normalized vectors
in $\mathbb{C}^3$ parameterized by $\theta$
\begin{eqnarray}
|\alpha_0\rangle &=& \sin^2\theta |0\rangle - \sin\theta \cos\theta
|2\rangle + \cos\theta |1\rangle, \nonumber\\
|\alpha_1\rangle &=& |0\rangle, \nonumber\\
|\alpha_2\rangle &=& \cos\theta |0\rangle + \sin\theta |2\rangle, \nonumber\\
|\alpha_3\rangle &=& \frac{1}{N} \left( \sin\theta \cos\theta |1\rangle +
\cos\theta |2\rangle \right), \nonumber\\
|\alpha_4\rangle &=& |1\rangle, \nonumber\\
&&{\rm with}\,\, N = \sqrt{\cos^2\theta + \sin^2\theta \cos^2\theta}.
\label{one-parameter}
\end{eqnarray}
These vectors follow the orthogonality relations shown in
Fig.~\ref{orthogonality_upb}.
Using these five vectors, we can construct a one parameter
family of UPB's (a subset of the six parameter family introduced
in~\cite{DiVincenzo2003}) in $\mathbb{C}^3\otimes \mathbb{C}^3$
given by the vectors:
\begin{equation}
\vert \alpha_j \rangle \otimes \vert \beta_j\rangle =
\vert \alpha_j \rangle \otimes \vert \alpha_{2j\; \bmod \;5}\rangle
\end{equation}
This one parameter family of UPB's parameterized by $\theta$
is local unitarily equivalent to Tiles UPB for
$\theta=3\pi/4$ and to the Pyramid UPB for $\theta=
\cos^{-1}{(\sqrt{5}-1})/2$, the two canonical
representatives of the six parameter family of UPBs in
$\mathbb{C}^3\otimes\mathbb{C}^3$~\cite{DiVincenzo2003}.  In
particular, for $\theta=\cos^{-1}{(\sqrt{5}-1})/2$, the
vectors $\vert \alpha_j\rangle$, after the application
of a suitable local unitary can be brought to the
standard Pyramid vectors given by:
\begin{equation} \label{pyramid vector}
\!|\psi_j\rangle\! =\!
\frac{\scriptstyle 2}{\scriptstyle \sqrt{5 + \sqrt{5}}}
\!\left[\!\cos\frac{2\pi j}{5}|0\rangle \!+\!
\sin\frac{2\pi j}{5}|1\rangle \!+\!
{\scriptstyle \frac{1}{2}\sqrt{1+\sqrt{5}}}\,|2\rangle\! \right].
\end{equation}

Consider the  KCBS vectors; the vectors corresponding
to the simplest form of
state dependent contextuality~\cite{Kcbs2008}:
\begin{eqnarray}
\label{kcbs_vectors}
|\phi_0\rangle &\! \equiv\! & \frac{\scriptstyle
1}{\scriptstyle \sqrt{N}}\left[
1, 0, \sqrt{\cos (\pi/5)} \;\right]^T, \nonumber \\
|\phi_{1, 4}\rangle &\! \equiv\!& \frac{\scriptstyle
1}{\scriptstyle \sqrt{N}}\left[ \cos(4\pi/5), \pm\sin(4\pi/5),
\sqrt{\cos(\pi/5)}\; \right]^T, \nonumber \\
|\phi_{2, 3}\rangle &\!  \equiv\!& \frac{\scriptstyle
1}{\scriptstyle \sqrt{N}}\left[ \cos(2\pi/5), \mp\sin(2\pi/5),
\sqrt{\cos(\pi/5)}\; \right]^T,
\end{eqnarray}
with $N = 1 + \cos(\pi/5)$.
\begin{obs}
The set of vectors defining the Pyramid UPB~\eqref{pyramid
vector} are identical to the KCBS vector
set~\eqref{kcbs_vectors}. In fact for every nonzero value of
$\theta$ the vectors $\vert \alpha_j\rangle$ provide contextual
vectors and they all have the orthogonality graph
	$C_5$ given in Fig.~\ref{orthogonality_upb}.
\end{obs}
Thus for each value of $\theta$ we thus have a set of contextual
vector $\vert \alpha_j\rangle $ and a corresponding UPB
$\vert \alpha_j\rangle \otimes
\vert \beta_j \rangle ,\,\, (j=0,\cdots 4)$, establishing a
connection between contextuality of the vectors and the corresponding UPB.

For any UPB, we have a corresponding bound entangled state
in its orthogonal complement.  Although all UPBs in
$\mathbb{C}^3 \otimes \mathbb{C}^3$ share identical
orthogonality graph as Fig.~\ref{orthogonality_upb}, the
corresponding bound entangled states differ in their Linear
Entropy of Entanglement (LEE).  For a bipartite state $\rho
\in \mathcal{H}_A \otimes \mathcal{H}_B$, the LEE is defined
as~\cite{Geza2015}
\begin{equation}
E_l(\rho) = \min_{\{p_i,|\psi_i\rangle\}}
\sum_i S_l(\operatorname{Tr}_B[|\psi_i\rangle\langle\psi_i|]),
\quad
S_l(\rho)=1-\operatorname{Tr}[\rho^2],
\end{equation}
where the minimization is over all pure-state decompositions
$\rho=\sum_i p_i |\psi_i\rangle\!\langle\psi_i|$.
In order to explain the variation in LEE, we consider the associated
\textit{orthogonal
representation} (OR) $\{\ket{\alpha_j}\}$ of the graph $C_5$ underlying each
UPB.
For a fixed OR $\{\ket{v_i}\}$ satisfying
$\langle v_i | v_j \rangle = 0$ whenever $i$ and $j$ are
\textit{adjacent in}
complement graph $\overline{G}$
(i.e., nonadjacent in $G$), we define the associated
\textit{contextual strength} $\mathcal{C}$ as
\begin{equation}
\mathcal{C}(G,\{v_i\}) = \max_{|\psi\rangle}
\sum_{i\in V} |\langle \psi | v_i\rangle|^2,
\end{equation}
where the maximization is over the handle vector $\ket{\psi}$.
Optimizing $\mathcal{C}(G,\{v_i\})$ further over all such valid ORs
$\{|v_i\rangle\}$ yields the Lov\'asz number
$\vartheta(G)$~\cite{Lovasz1979}:
\begin{equation}
\vartheta(G) =
\max_{\{|\psi\rangle, |v_i\rangle\}}
\sum_{i\in V} |\langle \psi | v_i\rangle|^2.
\end{equation}
	For $G=C_5$, $\vartheta(C_5)=\sqrt{5}$ gives the maximal
quantum violation of the KCBS inequality, whose classical
bound is given by the independence number $\alpha(G)=2$.
Graphs satisfying $\vartheta(G)>\alpha(G)$ are referred to
as \textit{quantum contextual graphs}(QCGs)~\cite{Tarrida2013}.  UPBs
constructed from vector sets $\{\ket{\alpha_j}\}$ with
$\mathcal{C}(C_5,\{\alpha_j\})$ approaching the quantum
optimum $\vartheta(C_5)=\sqrt{5}$ yield bound entangled
states of higher LEE, while those with smaller
$\mathcal{C}(C_5,\{\alpha_j\})$ exhibit weaker entanglement.
This establishes a quantitative correspondence between the
contextual strength of the underlying vectors and the
entanglement strength of the bound entangled state
associated to the UPB.
\begin{table}[t!]
\centering
\renewcommand{\arraystretch}{1.2}
\begin{tabular}{|c|c|c|c|}
\hline
$\theta$ & UPB Type &
$\mathcal{C}(C_5,\{\alpha_j(\theta)\})$ & LEE \\
\hline
$\cos^{-1}\!\frac{\sqrt{5}-1}{2}$ & Pyramid &
$\vartheta(C_5) = \sqrt{5}$  &
0.07295 \\
$3\pi/4$ & Tiles & 2.2287 & 0.06519 \\
$\pi/3$ &  & 2.2254 & 0.06335 \\
$\pi/6$ &  & 2.1641 & 0.01278 \\
$\pi/12$ & & 2.0590 & 0.00029 \\
\hline
\end{tabular}
\caption{
Contextual strength $\mathcal{C}(C_5,\{\alpha_j(\theta)\})$ and 
LEE
for various choices of the parameter $\theta$.  The
Pyramid ($\theta=\cos^{-1}{(\sqrt{5}-1})/2$)
and Tiles ($\theta=3\pi/4$) cases exhibit the
	highest contextual strengths and entanglement values.
}
\label{f_value}
\end{table}
For the set of vectors $\vert \alpha_j\rangle$ given
in Eqn.~\ref{one-parameter} we display the quantities
$\mathcal{C}(C_5,\{\alpha_j(\theta)\})$ LEE for various values of
the parameter $\theta$ in Table~\ref{f_value}. We observe a clear correlation
between the contextuality strength and LEE with Pyramid UPBs
having the highest value of LEE corresponding to the maximum
possible contextuality strength of $\sqrt{5}$.

Further, the correspondence between the KCBS vectors and
the Pyramid UPB vectors can be extended beyond the pentagon case.
In order to proceed, let us define the Generalized Pyramid
({\it GenPyramid}) UPB and the
corresponding Generalized KCBS vectors.
Consider a set of vectors in $\mathbb{C}^3$ :
\begin{eqnarray}
\label{genpyramid_vectors}
&&\vert {v_j}\rangle \equiv N_p\left[\cos\left({2\pi
j}/{p}\right), \,
\sin\left({2\pi j}/{p}\right),\, h_p\right], \nonumber \\
&&N_p=({1+|\cos{2 \pi t}/{p}|})^{-\frac{1}{2}}, \quad h_p =
({-\cos{2\pi t}/{p}})^{\frac{1}{2}}, \nonumber \\
&&t\,\,\mbox{(integer)} \quad  s.t. \quad \pi/2 \leq 2\pi
t/p \leq \pi, \nonumber \\
&&p=2m+1\,\, \mbox{(prime)}\quad \mbox{and}\quad j=0,\cdots
p-1.
\end{eqnarray}
The {\it GenPyramid} UPB for $m$ parties with Hilbert space
$\bigotimes_{k=1}^{m} \mathbb{C}^3$ is then constructed
from these vectors as:
\begin{equation}
\label{genpyramid_upb}
\vert {p_j}\rangle  = \vert{v_j}\rangle \otimes
     {\vert{v}}_{2j\!\!\!\!\mod p}\rangle \otimes
\cdots \otimes{\vert{v}}_{kj\!\!\!\!\mod p}\rangle, \,\, j = 0, \dots 2m.
\end{equation}
For \((m, t) = (2, 2)\), we obtain the
Pyramid UPB.
\begin{obs}
Although the original construction requires $p$ to be prime,
we find that valid UPBs also exist for certain non-prime dimensions.
In particular, for $(m,t)=(4,3)$ and $(m,t)=(12,10)$, the corresponding
values $p=9$ and $p=25$ yield UPBs.
A necessary condition for such non-prime cases appears to be that
$p$ must be an odd perfect square.
\end{obs}

To parallel the above generalization on the UPB side,
let us extend the KCBS construction to arbitrary odd-cycle scenarios.
The resulting family of \textit{Generalized KCBS vectors}
provides the natural contextual counterpart to the
{\it GenPyramid} UPB, with each set corresponding to an odd cycle
graph $C_n$.
Let us define the vectors:
\begin{eqnarray}
&&|u_j\rangle \equiv \left[\cos\phi,\; \sin\phi\; \cos\theta_j,
\;\sin\phi\;
\sin\theta_j\right]^T
\nonumber  \\
&&\mbox{where}\quad
\theta_j = {j \pi (n-1)}/{n}\quad
\mbox{for}\quad  j = 1, \cdots n
\nonumber \\
&&\cos^2\phi = \frac{\vartheta(C_n)}{n}
\quad \mbox{and} \quad
\vartheta(C_n) =  \frac{n\cos\frac{\pi}{n}}{1
+ \cos
\frac{\pi}{n} }.
\label{loor_cn}
\end{eqnarray}
The vectors $\{|u_j\rangle\}$
form a LOOR of the odd cycle graph $C_n$~\cite{Tarrida2013},
and the corresponding measurements
$M_j = 2|u_j\rangle\!\langle u_j| - \mathbb{I}$
achieve the maximal quantum violation of noncontextuality
inequalities for odd-$n$ cycle scenarios.
\begin{obs}
The {\it GenPyramid} UPB vectors~(\ref{genpyramid_vectors}) are
equivalent to Generalized KCBS vectors~(\ref{loor_cn}) when
$m = t$ and $n = 2m+1$. The equivalence
holds only in prime dimensions since a valid UPB exists only
when $2m+1$ is prime.
\end{obs}

\noindent
\textit{From Contextuality to UPB.--} 
We are now ready to extend the connection
beyond equivalences. In order to construct a class of UPB
originating from contextuality framework, we also define the corresponding LOORs
of $\overline{C}_n$, the graphs
complement to $C_n$. These representations will form the
second local component of the set of product 
vectors in our
UPB. As demonstrated
in~\cite{Tarrida2013}, a LOOR of $\overline{C}_n$ in
dimension $n-2$ is realized by the vectors $\{v_{j,k}\}
\subset \mathbb{C}^{n-2}$ for $0 \le j \le n-1$ and $0 \le k
\le n-3$ defined as:
\begin{equation}
v_{j, k} =
\begin{cases}
	\sqrt{{\vartheta(\overline{C}_n)}/{n}} & \text{for } k
= 0, \\
T_{j, m} \cos{R_{j,m}} & \text{for } k = 2m - 1, \\
T_{j, m} \sin{R_{j,m}} & \text{for } k = 2m,
\end{cases}
\label{loor_bar_cn}
\end{equation}
where
\begin{align*}
&m =1\cdots\frac{n\!-\!3}{2},
\vartheta(\overline{C}_n) = \frac{1 +
\cos\frac{\pi}{n}}
{\cos\frac{\pi}{n}},
R_{j,m} = \frac{j(m+1)\pi}{n},\\
&T_{j,m} = (-1)^{j(m+1)} \sqrt{ \frac{2 \left( \cos
\tfrac{\pi}{n}
+ (-1)^{m+1} \cos\tfrac{(m+1)\pi}{n}
\right)}{n \cos
\tfrac{\pi}{n} } }.
\end{align*}
This LOOR of $\overline{C}_n$, together
with the LOOR of $C_n$ in $\mathbb{C}^3$, yield product
vectors in $\mathbb{C}^3 \otimes \mathbb{C}^{n-2}$ forming a
UPB. Since the construction is based on
orthogonal representations of cycle graph and its complement
which are QCGs,
we refer to the resulting UPB as \textit{GenContextual} UPB.
\begin{thm}
The product of LOORs of odd cycle graphs~(\ref{loor_cn})
and their complements~(\ref{loor_bar_cn}) given by vectors,
$|\phi_j\rangle= |u_j\rangle \otimes |v_j\rangle$ constitute a class
of minimal UPBs in $\mathbb{C}^3 \otimes \mathbb{C}^{n-2}$
for odd $n$.
\label{th1}
\end{thm}
\begin{proof}
Let $G =(V,E)$ be the orthogonality graph of a set of
vectors $\{|\varphi^{(1)}\rangle, \dots,
|\varphi^{(k)}\rangle\} \subset \mathbb{C}^d$,
where vertices represent state vectors and an edge exists
between two vertices if they are orthogonal.
A subset of vertices $W \subseteq
        V $ is said to be \textit{saturated} if the
corresponding vectors span $\mathbb{C}^d$; otherwise, it is
\textit{unsaturated}.
\begin{lem}[Sufficient and Necessary Conditions for a UPB~{\cite{Fei2023}}]
\label{lemma:suf_nec}
Let $\mathcal{U}$ be a set of $k$ product vectors in $\mathbb{C}^{d_1} \otimes
\mathbb{C}^{d_2} \otimes \cdots \otimes \mathbb{C}^{d_N}$, and let
$(G_m)_{m=1}^N$ denote the corresponding orthogonality graphs. Then
$\mathcal{U}$ forms a UPB if and only if:
\begin{enumerate}
\item $\bigcup_{m=1}^N G_m = K_k$;
\item For every $N$-tuple $(W_1, W_2, \ldots, W_N)$, where each $W_m$ is an
unsaturated set in $G_m$, we have $\bigcup_{m=1}^N W_m \neq V$.
\end{enumerate}
\end{lem}
In our case, the orthogonality graphs are defined as
follows: \( G_A = C_n \) (the cycle graph), associated with
the local vectors \( \{|u_j\rangle\} \subset \mathbb{C}^3
\), and \( G_B = \overline{C}_n \) (the complement of the
cycle graph), associated with \( \{|v_j\rangle\} \subset
\mathbb{C}^{n-2} \).

The first condition is met as $C_n \cup \overline{C}_n = K_n.$
The second condition requires that for every pair of subsets
of vertices
$(W_A, W_B)$, with $W_A \subseteq V(G_A)$ and $W_B \subseteq
V(G_B)$ being unsaturated in their respective graphs, the
union $W_A \cup W_B$ does not cover all of $V$.

For $G_A = C_n$, the LOOR in $\mathbb{C}^3$ ensures that any
three or more orthogonal vectors span $\mathbb{C}^3$.
Therefore, a subset $W_A$ is unsaturated only if $|W_A| \leq
2$. For $G_B = \overline{C}_n$, the LOOR in $\mathbb{C}^{n-2}$
implies that any set of $n - 2$ or more orthogonal vectors
spans $\mathbb{C}^{n-2}$. Hence, $W_B$ is unsaturated only
if $|W_B| \leq n - 3$.

Now, for any such pair of unsaturated sets $(W_A, W_B)$,
their union satisfies~\cite{Fei2023}: \[ |W_A \cup W_B| \leq |W_A| + |W_B|
\leq 2 + (n - 3) = n - 1 < n, \] so $\bigcup_{m= A,B} W_m
\neq V$, and thus does not cover $V$. \\
	Since both the required conditions are satisfied, the set
$\{|\phi_j\rangle\}$ forms a UPB.
\end{proof}

For an $m$-partite Hilbert space $\mathbb{C}^d =
\bigotimes_{k=1}^m \mathbb{C}^{d_k}$, the cardinality $n$ (number of vectors) of
any UPB satisfies the lower bound~\cite{Bennett1999b}
\begin{equation}
n \geq \sum_{i=1}^m (d_i - 1) + 1.
\end{equation}
Alon and Lov\'asz later showed that when all local dimensions $d_i$ are odd,
this bound is saturated~\cite{Alon2001}, and such UPBs are referred to as
\textit{minimal}.
The \textit{GenContextual} UPB in $\mathbb{C}^3 \otimes \mathbb{C}^{d_2}$ with
odd $d_2$ satisfies this condition, having cardinality $n = d_2 + 2$, and
therefore constitutes a minimal UPB.

The most notable feature of the \textit{GenContextual} UPB is that it serves as
the canonical orthogonality graph for all minimal UPBs in $\mathbb{C}^3 \otimes
\mathbb{C}^{n-2}$ when $n$ is odd. Any minimal UPB necessarily in $\mathbb{C}^3
\otimes \mathbb{C}^{n-2}$ for odd $n$, exhibits the same orthogonality
structure as \textit{GenContextual}.

\begin{defn}[\textbf{Graph Equivalence of
UPBs}~\cite{Lovasz2009}] 
\label{def:Graph_Equivalence} Let $\mathbb{C}^d =
\bigotimes_{k=1}^m \mathbb{C}^{d_k}$ be an $m$ partite
Hilbert space, and let $A$ and $B$ be two UPBs in
$\mathbb{C}^d$ with orthogonality graphs $(A, A \times A,
c_1)$ and $(B, B \times B, c_2)$, respectively. We say that
$A$ and $B$ are \textit{graph equivalent}, denoted $A \sim B$,
if there exists a permutation $\Pi : A \to B$ such that
\[
c_1(u, v) = c_2(\Pi(u), \Pi(v))
\quad \text{for all } u, v \in A.
\]
Here, $c_i : S \times S \to \mathcal{P}(\{1, 2, \dots, m\})$
(for $i = 1,2$) are
the colorings of the edges defined by
\[
c_i(|\psi\rangle, |\phi\rangle) = \left\{ k \,\middle|\,
|\psi^{(k)}\rangle
\perp |\phi^{(k)}\rangle \right\}.
\]
\end{defn}

\begin{lem}\label{graph_equivalence_lemma}
Let $G$ be a 2 connected, 2 regular graph with $n \geq 3$ vertices. Then $G$ is
an $n$-cycle.
        \label{lem2}
\end{lem}

\begin{proof}
A finite 2-regular graph is a disjoint union of cycles. Since $G$ is
2-connected, it is connected and remains connected after removal of any single
vertex. A disjoint union of multiple cycles is disconnected, which contradicts
the 2-connectedness assumption. Therefore, $G$ must consist of exactly one
cycle, i.e., $G \cong C_n$.
\end{proof}

\begin{thm}\label{th:main}
Let $\mathbb{C}^d = \mathbb{C}^3 \otimes \mathbb{C}^{\,d_2}$ be a bipartite
Hilbert space with $d_2$ odd.
If $A$ and $B$ are minimal UPBs in $\mathbb{C}^d$, each of cardinality $n = d_2
+ 2$,
then they are graph equivalent.
Consequently, every minimal UPB in $\mathbb{C}^3 \otimes \mathbb{C}^{\,d_2}$
is graph equivalent to the \textit{GenContextual} UPB.
\end{thm}

\begin{proof}
Let $A$ and $B$ be minimal UPBs in $\mathbb{C}^3 \otimes
\mathbb{C}^{d_2}$, each of cardinality $n = d_2 +
2$. Consider the complete graph $K_n$ whose vertices
represent the states. Color an edge green if the
corresponding states are orthogonal on party 1
($\mathbb{C}^3$), and red if they are orthogonal on
party 2 ($\mathbb{C}^{d_2}$). By the properties of
minimal UPBs, every edge of $K_n$ receives exactly
one color, partitioning the edge set into two
factors.

The green factor (party-1 orthogonality graph) is 2-regular
and connected, hence by
Lemma~\ref{graph_equivalence_lemma} it is an
$n$-cycle. Since $n$ is odd, this is an odd cycle.
The red factor (party-2 orthogonality graph) is the
complement of the green factor in $K_n$.

Choose a permutation $\Pi: A \to B$ that maps the party-1
cycle of $A$ to that of $B$. This permutation
preserves the edge coloring: if $u,v \in A$ are
orthogonal on party 1 (green edge), then
$\Pi(u),\Pi(v)$ are adjacent in $B$'s cycle and thus
orthogonal on party 1; if orthogonal on party 2 (red
edge), they are non-adjacent in $A$'s cycle, so
$\Pi(u),\Pi(v)$ are non-adjacent in $B$'s cycle and
thus orthogonal on party 2. Therefore $c_1(u,v) =
c_2(\Pi(u),\Pi(v))$ for all $u,v \in A$, and $A \sim
B$.
\end{proof}
It is well established that all UPBs are indistinguishable
under LOCC giving rise to \textit{nonlocality without
entanglement}. However, certain UPBs remain
indistinguishable under even more general measurement
classes.  Cohen conjectured that all known minimal UPBs are
indistinguishable under the topological closure of LOCC,
denoted by $\overline{\text{LOCC}}$~\cite{Cohen2023}.
Specific examples such as the Tiles, Shifts, and {\it
GenTiles2} UPBs have been shown to be indistinguishable
under $\overline{\text{LOCC}}$~\cite{Cohen2022}. However,
these UPBs can be distinguished by separable (SEP)
measurements, a class strictly larger  than
$\overline{\text{LOCC}}$. Although most UPBs are
distinguishable under SEP, the Feng UPB in $\mathbb{C}^4
\otimes \mathbb{C}^4$ is a known exception and remains
SEP-indistinguishable~\cite{Bandyopadhyay2015}. An even
larger class corresponds to measurements that are positive
under partial transpose (PPT), under which all UPBs can be
perfectly distinguished. Therefore we have the following
hierarchy:
\begin{equation}
\text{LOCC} \subsetneq \overline{\text{LOCC}} \subsetneq
\text{SEP} \subsetneq
\text{PPT},
\end{equation}
In $\mathbb{C}^3
\otimes \mathbb{C}^3$, all UPBs are known to be
distinguishable by SEP measurements~\cite{DiVincenzo2003}.
Using semidefinite programming (SDP), one can verify that
both the \textit{GenTiles2} and \textit{GenContextual}
UPBs are SEP-distinguishable, following the standard
SDP framework for analyzing distinguishability of 
product bases~\cite{Bandyopadhyay2015}.
We conjecture, however, that the \textit{GenContextual} UPB,
being minimal, remains indistinguishable under
$\overline{\text{LOCC}}$ measurements.

\noindent
\textit{From UPB to contextuality.--}
So far  we established how contextual graphs can give rise
to UPBs.  We now examine the reverse direction and examine
how a given UPB can reveal an underlying contextual
structure.  In particular, we focus on the UPB introduced by
DiVincenzo \textit{et al.}~\cite{DiVincenzo2003}, known as
the \textit{QuadRes} UPB.

\begin{defn}[\textbf{Quadric Residues}~\cite{DiVincenzo2003}]
 Let $\mathbb{Z}_p$ be the group of integers
modulo $p$ and  $\mathbb{Z}_p^* = \mathbb{Z}_p \setminus \{0\}$. Then
the quadratic residues $Q_p$ are defined as:
\[
Q_p = \{q \in \mathbb{Z}_p^* \;\; \mbox{s.t.}~~\exists
~~x \in \mathbb{Z}. q \equiv x^2
\!\!\!\! \mod p\}
\]
and they forms a group under multiplication modulo $p$.
\end{defn}

Consider, for any given $d \in \mathbb{N}$ and $x \in
\mathbb{Z}_{2d-1}^* \setminus Q_{2d-1}$ such that,  $p =
2d-1 = 4t-1$  is prime for some \textit{t}. A function $Q :
\mathbb{Z}_p \to \mathbb{C}^d$ defines a set of states
\begin{equation}
\label{paley_vectors}
|Q(a)\rangle = \sqrt{N} |0\rangle + \sum_{q \in Q_p}
e^{i2\pi qa/p}
|2q\rangle
\end{equation}
where $N = \max(-\sum_{q \in Q_p} e^{i2\pi q/p}, 1 + \sum_{q
\in Q_p} e^{i2\pi q/p})$ is the normalization factor.  The
QuadRes UPB states are then defined as
\begin{equation}
|\psi_i\rangle =
|Q(i)\rangle \otimes |Q(ix)\rangle,
\end{equation}
for all $i \in \mathbb{Z}_p$. 
The vectors given in Equation~(\ref{paley_vectors})
corresponds to the LOOR of Paley graphs in prime dimensions,
establishing a deeper link between contextuality and the
structure of UPBs. Paley graphs in prime dimensions are
self complementary vertex transitive graphs. These
properties make the calculation of Lov\'asz number of such
Paley graphs simple.

As discussed in~\cite{Lovasz1979}, for any graph $G$ on $n$
vertices, the Lovász theta number of $G$ and its complement
$\overline{G}$ satisfy the inequality:
\begin{equation}
\vartheta(G)\, \vartheta(\overline{G}) \geq n.
\end{equation}
Moreover, if $G$ is vertex transitive, then equality holds:
\begin{equation}
\vartheta(G)\, \vartheta(\overline{G}) = n.
\end{equation}
\begin{obs}\label{theta_paley_graphs}
Let $P_q$ be a Paley graph. Then $\vartheta(P_q) = \sqrt{q}$.
\end{obs}
Hence, for every $q$, we know the value of $\vartheta(P_q)$.
As previously mentioned, in~\cite{Broere1988} it was proved
that, if $q$ is a square, then $\alpha(P_q) = \sqrt{q}$.
Thus, in this case $\vartheta(P_q)$ and $\alpha(P_q)$
coincide.  Moreover, when $q$ is a prime, $P_q$ is circulant
and the subgraph of $P_q$ induced by the neighborhood of the
vertex $0$ is also circulant~\cite{Magsino2019}.  The role
of circulant Paley graphs was previously explored in a
limited way in the context of the global exclusivity
principle~\cite{Cabello2013}.  The Paley graphs are also
noteworthy because they exhibit a higher degree of quantum
contextuality  compared to the odd cycle graphs in terms of
the ratio  $\vartheta(G)/\alpha(G)$. One can therefore
construct robust noncontextuality inequalities using paley
graphs which could be tested
experimentally~\cite{Cabello2008a,Cabello2008b,Marques2014}.
\begin{table}[h!]
\centering
\renewcommand{\arraystretch}{1.2}
\begin{tabular}{|c|c|c|c|c|c|c|}
\hline
$q$ & 5 & ~~9~~ & 13 & 17 & 25 & 29 \\
\hline
$\vartheta(P_q)$& $\sqrt{5}$ & 3 & $\sqrt{13}$ & $\sqrt{17}$
& 5 & $\sqrt{29}$ \\
\hline
$\alpha(P_q)$ &2 &3 &3 &3 &5 &4 \\
\hline
\end{tabular}
\caption{The Lov\'asz number
and independence number for some  Paley graphs $P_q$.
	\label{table_2}}
\end{table}
The Lov\'asz number and the independence number 
for Paley graphs with different $q$ values are shown in Table~\ref{table_2}.
The ration $\vartheta(G)/\alpha(G)$ is an important
indicator of the degree of contextuality.

\noindent
\textit{Conclusions.--}
In this work we have shown
that LOORs associated with quantum contextual graphs can be
used to construct a class of UPB and that certain class of
UPBs are naturally related to quantum contextual graphs,
bringing out intriguing new connections between quantum
contextuality and UPBs.  These findings extend earlier work
linking UPBs to Bell inequalities with no quantum violation,
which arise from the local orthogonality
principle~\cite{Sainz2014}, itself a special case of the
exclusivity principle Since
the contextuality framework employed here is also based on
the exclusivity principle, our results indicate a broader connection
between exclusivity-based constraints and the structure of
UPBs.

Several open questions remain. Notably, the precise
conditions under which the LOOR-UPB correspondence holds is
yet to be established, as many known UPBs appear unrelated
to contextuality.  Our work suggests several avenues for
future research. For instance, one could try to extend the
\textit{GenContextual} UPB construction to other dimensions,
following the framework of Yang
\textit{et~al.}~\cite{Yang2015}. It would also be worthwhile
to examine whether \textit{GenContextual} UPB can be
perfectly discriminated by LOCC using a maximally entangled
resource of dimension $\lceil \frac{m}{2}\rceil\otimes
\lceil \frac{m}{2}\rceil$ , as is the case for the {\it
GenTiles2} UPB~\cite{Cohen2008}, or if a higher degree of
entanglement is required.  While our analysis has focused on
odd dimensional systems, extending this framework to even
dimensional cases and to multipartite UPBs beyond three
parties remains an open problem.  Additionally, the role of
\textit{general orthogonal representations} and the associated
\textit{contextuality strength} in characterizing UPBs and
their entanglement properties warrants further
investigation.

\noindent
\textit{Acknowledgements.---}Gurvir acknowledges
Jaskaran Singh and A. J. Lo\'pez-Tarrida for
fruitful discussions.
%
\end{document}